\documentclass[12pt, reqno]{amsart}
\usepackage{graphics}
\usepackage{graphicx}
\usepackage{amsmath, amssymb, amsfonts, amsbsy, amsthm, latexsym,color}
\everymath{\displaystyle}
\usepackage{enumerate}
\newtheorem{theorem}{Theorem}[section]
\newtheorem{proposition}[theorem]{Proposition}
\newtheorem{lemma}[theorem]{Lemma}

\theoremstyle{definition}
\newtheorem{definition}[theorem]{Definition}
\newtheorem{example}[theorem]{Example}

\newtheorem{remark}[theorem]{Remark}

\newcommand{\C}{{\mathbb{C}}}
\newcommand{\F}{{\mathbb{F}}}
\newcommand{\Sb}{{\mathbb{S}}}

\newcommand{\vtu}{{\tilde{\bf u}}}

\newcommand{\CHI}{\hbox{\raise .4ex \hbox{$\chi$}}}

\newcommand{\R}{{\mathbb{R}}}

\newcommand{\ls}{$\ell_1$-synthesis}
\newcommand{\la}{$\ell_1$-analysis}
\newcommand{\PDe}{$(\text{P}_{\vD,\epsilon})$}
\newcommand{\pdz}{$(\text{P}_{\vD,0})$}

\def\va{{\bf a}} \def\vb{{\bf b}}  \def\vd{{\bf d}}
\def\ve{{\bf e}} \def\vf{{\bf f}} \def\vg{{\bf g}} \def\vh{{\bf h}}

\def\vu{{\bf u}} \def\vv{{\bf v}} \def\vw{{\bf w}} \def\vx{{\bf x}}
\def\vy{{\bf y}} \def\vz{{\bf z}}

\def\vA{{\bf A}} \def\vB{{\bf B}}  \def\vD{{\bf D}}
 \def\vF{{\bf F}} \def\vG{{\bf G}} 
\def\vI{{\bf I}}   
\def\vM{{\bf M}} \def\vN{{\bf N}}

\newcommand{\supp}{\mathrm{supp}}
\newcommand{\spark}{\mathrm{spark}}

\newcommand{\nsp}{$\text{NSP}$}
\newcommand{\snsp}{{\text{SNSP}}}

\newcommand{\sgn}{\text{sgn}}

\newcommand{\nspD}{$\vD$-NSP}
\newcommand{\snspD}{$\vD$-SNSP}

\begin{document}

\title[$\ell_1$-synthesis method in compressed sensing]{A null space analysis of the $\ell_1$-synthesis method in dictionary-based compressed sensing}

\author{Xuemei Chen}
\address{Department of Mathematics,
University of Maryland, College Park, MD 20742}
\email{xuemeic@math.umd.edu}

\author{Haichao Wang}
\address{Department of Mathematics,
University of California, Davis, CA 95616}
\email{hchwang@ucdavis.edu}

\author{Rongrong Wang}
\address{Department of Mathematics,
University of Maryland, College Park, MD 20742}
\email{rongwang@math.umd.edu}


\date{\today}


\begin{abstract}
An interesting topic in compressed sensing aims to recover signals with sparse representations in a dictionary. Recently the performance of the $\ell_1$-analysis method has been a focus, while some fundamental problems for the $\ell_1$-synthesis method are still unsolved. 
For example, what are the conditions for it to stably recover compressible signals under noise? Whether  coherent dictionaries allow the existence of sensing matrices that guarantee  good performances of the \ls\ method?
To answer these questions, we build up a framework for the $\ell_1$-synthesis method. In particular, we propose a dictionary-based null space property (\nspD) which, to the best of our knowledge, is the first sufficient and necessary condition for the success of $\ell_1$-synthesis without measurement noise. 
 With this new property, we show that when the dictionary $\vD$ is full spark,  it cannot be too coherent otherwise the \ls\ method  fails for all sensing matrices. We also prove that in the real case, \nspD\ is equivalent to the stability of \ls\ under noise.

\end{abstract}
\maketitle

\section{Introduction}
Compressed sensing addresses the problem of recovering a sparse signal $\vz_0\in \F^d$ ($\F=\C$ or $\R$) from its undersampled and corrupted linear measurements $\vy=\vA \vz_0+\vw\in\F^m$, where $\vw$ is the noise vector such that $\|\vw\|_2\leq\epsilon$. The number of measurements $m$ is usually much less than the ambient dimension $d$, which makes the problem ill-posed in general. A vector is said to be \emph{$s$-sparse} if it has at most $s$ nonzero entries. The sparsity of $\vz_0$ makes the reconstruction possible.  The following optimization algorithm, also known as the Basis Pursuit, can reconstruct $\vz_0$ efficiently from the perturbed observation $\vy$ \cite{Stab,LqSim}:
\begin{equation}\label{equ_Pe}
\hat{\vz}=\arg\min_{\vz\in \F^d} \|\vz\|_1, \quad\text{ s.t. } \|\vy-\vA \vz\|_2\leq \epsilon .
\end{equation}

A primary task of compressed sensing is to choose appropriate sensing matrix $\vA$ in order to achieve good performance of \eqref{equ_Pe}. Candes and Tao proposed the \emph{restricted isometry property (RIP)}, and show that it provides stable reconstruction of approximately sparse signals via \eqref{equ_Pe} \cite{Near}. Moreover, many random matrices satisfy RIP with high probability \cite{Decode, RV08}.

 Another well-known condition on the measurement matrix is the null space property. A matrix $\vA$ is said to have the \emph{Null Space Property of order $s$} ($s$-\nsp) if
\begin{equation}\label{eq:NSP}
\forall \vv\in \ker \vA \backslash \{0\}, \ \forall |T|\leq s,\ \ \ \|\vv_T\|_1<\|\vv_{T^c}\|_1,
\end{equation}
where $|T|$ is the cardinality for the index set $T\in\{1,2,\dots,d\}$, $T^c$ is its complementary index set and $\vv_T$ is the restriction of $\vv$ on $T$. NSP is known to characterize the exact reconstruction of all $s$-sparse vectors via \eqref{equ_Pe} when there is no noise ($\epsilon=0$) \cite{Dono,SparseDecom}. It has also been proven that the NSP matrices admit a similar stability result as RIP except that the constants may be larger~{\cite{ACP11, Sun11}.

In all of the above discussions, it is assumed that the signal $\vz_0$ is sparse with respect to an orthonormal basis.  A recent direction of interest in compressed sensing concerns problems where signals
are sparse in an overcomplete dictionary  $\vD$ instead of a basis,
see \cite{DMP1, Dic_Candes, LML12, ACP11, DNW}.
Here $\vD$ is a $d \times n$ matrix with full column rank.  We also call $\vD$ a frame in the sense that the columns of $\vD$ form a finite frame. \emph{A finite frame for $\F^d$} is a finite collection of vectors that span $\F^d$. We refer interested readers to \cite{frame} for a background on frame theory.

In this setting, the signal $\vz_0 \in \F^d$ can be represented as $\vz_0=\vD\vx_0,$ where
$\vx_0$ is an $s$-sparse vector in $\F^n$.
We refer to such signals as \emph{dictionary-sparse signals} or \emph{frame-sparse signals}.  When the dictionary $\vD$ is specified, we also call them as \emph{$\vD$-sparse signals}.
We refer the problem of recovering such $\vz_0$ from the linear measurement $\vy=\vA\vz_0$ as \emph{dictionary-based compressed sensing}, and the ordinary compressed sensing problem as \emph{basis-based compressed sensing}.

A natual way to obtain a good approximation $\hat\vz$ of $\vz_0$  is to use the following approach

\begin{minipage}[t]{0.1\textwidth}
 $$(\text{P}_{\vD})$$
 \end{minipage}
\begin{minipage}[t]{0.87\textwidth}
 \begin{align} \label{equ_PD}
& \hat\vx=\arg\min \|\vx\|_1 \quad\text{ s.t. } \vA\vD\vx =\vy ,\\\label{equ_PD2}
&\hat\vz=\vD\hat\vx.
\end{align}
\end{minipage}
The above method is called the $\ell_1$-synthesis or synthesis based method~\cite{LML12, DMP1} due to the second synthesizing step. In the case when the measurements are perturbed, we naturally solve the following:

\begin{minipage}[t]{0.1\textwidth}
 $$(\text{P}_{\vD,\epsilon})$$
 \end{minipage}
\begin{minipage}[t]{0.87\textwidth}
 \begin{align*} 
& \hat\vx=\arg\min \|\vx\|_1 \quad\text{ s.t. } \|\vA\vD\vx -\vy\|_2\leq\epsilon ,\\\label{equ_PD2}
&\hat\vz=\vD\hat\vx.
\end{align*}
\end{minipage}

The frame-based compressed sensing is motivated by the widespread use of overcomplete dictionaries and frames in signal processing and data analysis. Many signals naturally possess sparse frame coefficients, such as radar images (Gabor frames ~\cite{PRT12,HS09,SW13}), cartoon like images (curvelets \cite{k13}), images with directional features (shearlets \cite{cd12}), and etc. Other useful frames include wavelet frames \cite{ds10} and harmonic frames. If the underlying frame is unknown but training data is available, the frame may also be constructed or approximated by learning. The greater flexibility and stability of frames make them preferable for practical purposes to achieve greater accuracy under imperfect measurements.

Despite the countless application of  frame-sparse signals, the compressed sensing literature is still lacking on this subject, especially on the issue whether the frame $\vD$ can be allowed to be highly coherent or not. 
Coherence is a quantity that measures the correlation between frame vectors. When all the columns $\{\vd_j\}$ of $\vD$ are normalized, its \emph{coherence} is defined as $$\mu(\vD)=\max_{i\neq j}|\langle \vd_i, \vd_j\rangle|.$$ A highly coherent $\vD$ is a frame with big coherence.

The work in \cite{DMP1} establishes conditions on $\vA$ and $\vD$ to make the compound $\vA\vD$ satisfy RIP. However, as has been pointed out in \cite{Dic_Candes, LML12}, forcing $\vA\vD$ to satisfy RIP or even the weaker property NSP implies the exact recovery of both $\vz_0$ and $\vx_0$, which is unnecessary if we only care about obtaining a good estimate of $\vz_0$. In particular, they argue that if $\vD$ is perfectly correlated (has two identical columns), then there are infinitely many minimizers of \eqref{equ_PD}, but all of them lead to the true signal $\vz_0$ after \eqref{equ_PD2} .

The work in \cite{Dic_Candes} proposes the \la\ method:
\begin{equation}\label{equ_ana}
\hat{\vz}=\arg\min_{\vz\in R^d} \|\vD^*\vz\|_1, \quad\text{ s.t. } \|\vy-\vA \vz\|_2\leq \epsilon .
\end{equation}

It is proved that if $\vA$ satisfies a dictionary related RIP condition (DRIP), then the reconstruction is stable under the assumption that $\vD^* \vz_0$ is sparse. In the case that $\vD$ is a Parseval frame, $\vD^*\vz$ is the frame coefficients of $\vz$ in the canonical dual. In the finite frame setting, the columns of $\vD$ form a \emph{Parseval frame} if $\vD\vD^*=\vI$, and $\vF$ is a \emph{dual frame} of $\vD$ if $\vF\vD^*=\vD\vF^*=\vI$.

The work \cite{Dic_Candes} is the first result of compressed sensing that does not require a dictionary to be highly incoherent. But it requires the sparsity of $\vD^*\vz_0$, which does not seem to fit into the original setting very well. The work in \cite{LML12} proposes an optimal dual based $\ell_1$-analysis approach along with an efficient algorithm, but the stability result does not hold universally for all frame-sparse signals.

All of the work mentioned above take the Basis Pursuit approach. The work in \cite{DNW} takes  a greedy algorithm approach to solve the frame-based compressed sensing problem. They show that  this greedy algorithm will recover the signal accurately under the DRIP condition.

This paper aims to build up a framework for the frame-based compressed sensing, which the literature is lacking. We focus on the \ls\ method for various reasons. We introduce  new conditions based on the null space of $\vA$ and $\vD$, which will guarantee the stable recovery via the \ls\ method.
Through this condition, we partially address the question whether a frame-sparse signal can be accurately recovered via the $\ell_1$-synthesis approach with a highly coherent frame.

\subsection{Overview and main results}
The first main result of this paper is to provide a framework of the frame-based compressed sensing, where we try to answer some basic questions of this subject.
Section~\ref{sec_basic} studies the question of what is the minimum condition on a sensing matrix such that a decoder exists.
Section~\ref{sec_why} explains why the \ls\ method is a reasonable algorithm.
Section~\ref{sec:nsp} explores what is the condition such that this algorithm will recover frame-sparse signals from noiseless measurements, and what kind of matrices $\vA$ will satisfy this condition (\nspD).
Section \ref{sec_sta} further shows how this condition performs under noisy measurements.
Section \ref{sec_gap} studies how  this condition compares to previously known conditions such as $\vA\vD$ having NSP.
To the best of our knowledge, these results are the first characterization of frame-based compressed sensing via the $\ell_1$-synthesis approach.


The second main result is related to the question  whether the \ls\  approach allows a highly coherent frame. We prove that under the assumption that $\vD$ is full spark, $\vA$ having \nspD\ is equivalent to $\vA\vD$ having NSP, which requires $\vD$ itself to have NSP. 
As a consequence, since full spark frames is a set of probability measure 1, we generally have to use an incoherent frame if we expect the \ls\ method to perform well. Section \ref{sec_coh} provides an error bound on the reconstructed signals when an inadmissible frame (see Definition \ref{def_adm}) is used. Some simulations are done in Section~\ref{sec_sim} to illustrate related theorems.

\subsection{Notations and setup}
We use boldface lowercase letters to denote vectors and boldface uppercase letters to denote matrices. The signals that we are interested in recovering  live in $\mathbb{F}^d$ with $\mathbb F$ being $\C$ or $\R$. All the results are for both complex and real cases except Theorem \ref{thm_equiv}.

The frames that we consider have $n$ frame vectors. By slight abuse of notations, given a frame $\vD$, we also use $\vD$ as its synthesis operator, which is a $d\times n$ matrix. So the set of frames of this size can be denoted as
$$\F^{d\times n}=\{\F\text{ valued matrices of size }d\times n\}.$$
Given a frame $\vD\in\F^{d\times n}$,   $\vD^{-1}(E)$  denotes the preimage of the set $E$ under the operator $\vD$, $\vD^*$ is the conjugate transpose of $\vD$, $\vD_T$ is the submatrix of $\vD$ formed by taking columns corresponding to the index set $T$, and $\nu_{\vD}$ is the smallest positive singular value of $\vD$.

We define
$$\vD\Sigma_s=\{\vz\in\F^d: \exists\ \vx, \text{ such that } \vz=\vD\vx, \|\vx\|_0\leq s\}.$$ In this paper, we are interested in recovering the signal $\vz_0$ that is  $\vD$-sparse (in the set of $\vD\Sigma_s$) or $\vD$-compressible(has small distance from the set $\vD\Sigma_s$). The term $\sigma_s(\vw)=\min_{\vv\in\Sigma_s}\|\vw-\vv\|_1$ denotes the $\ell_1$ residue of the best $s$-term approximation to $\vw$.

For $p>0$, we use $\|\vx\|_p$ to denote the regular $\ell_p$ norm of a vector $\vx$ (it is quasi-norm in the case that $p<1$). When $p=\infty$, it is the largest component in magnitude. When $p=0$, it is defined as the cardinality of the nonzero components of $\vx$. We also use $\|\vx\|_{\min}$ to represent the smallest nonzero component of $\vx$ in magnitude.


\section{Starting from the basic: $\ell_0$-minimization}\label{sec_basic}


This section deals with the noiseless case, and provides characterizations on when the measurement vector $\vy$ determines a unique signal $\vz_0$. To extract the information that $\vy$ holds about $\vz_0\in \vD\Sigma_s$, we use a decoder $\Delta$ which is a mapping from $\F^m$ to $\F^d$. Thus, $\Delta(\vy)=\Delta(\vA\vz_0)$ is the reconstructed signal.
Since we are only able to see the output $\vy = \vA\vz$, to have any hope of recovering the true signal at all, we require $\vA$ to be injective on the set $\vD\Sigma_{s}$. One can show that this injectivity condition is equivalent to
\begin{equation}\label{equ_dnsp0}
\ker \vA\cap \vD\Sigma_{2s}=\{0\}.
\end{equation}

Once this necessary condition is satisfied,  the following minimization scheme

\begin{minipage}[t]{0.1\textwidth}
 $$(\text{P}_{\vD,0})$$
 \end{minipage}
\begin{minipage}[t]{0.87\textwidth}
 \begin{align*} 
& \hat\vx=\arg\min \|\vx\|_0 \quad\text{ s.t. } \vA\vD\vx =\vy ,\\
&\hat\vz=\vD\hat\vx.
\end{align*}
\end{minipage}
is a perfect decoder in the sense that it has the unique solution $\hat\vz=\vz_0$. Indeed, suppose $\hat \vx$ is any minimizer, therefore $\vA\vD\hat \vx=\vA\vD\vx_0$. By injectivity of $\vA$ on $\vD\Sigma_s$, we get $\hat \vz=\vD\hat \vx=\vD\vx_0=\vz_0$.

 We summarize the above as the following proposition, which can be viewed as a generalization of Lemma 3.1 in \cite{kterm}.

\begin{proposition}\label{pro_basic}
The following conditions are equivalent:
\begin{enumerate}
\item[(a)] For any $\vz_0\in \vD\Sigma_s$, there exists a decoder $\Delta$ such that $\Delta(\vA\vz_0)=\vz_0$.
\item [(b)]$\ker \vA\cap \vD\Sigma_{2s}=\{0\}$.
\item [(c)] The problem \pdz\ has a unique solution $\hat\vz$ and $\hat\vz=\vz_0$.
\item[(d)] $\rm{rank}\ \vD_T=\rm{rank}\ \vA\vD_T$, for any index $|T|\leq 2s$.
\end{enumerate}
\end{proposition}

\proof The equivalence of (a)(b)(c) has been discussed as above. We only need to prove the equivalence of (d) with others.

(b)$\Rightarrow$(d) It suffices to show that $\ker \vA\vD_T\subset \ker \vD_T$. Suppose $\vA\vD_T\vx'=0$, if we let $\vx\in\C^n$ be the $2s$-sparse vector that equals to $\vx'$ on $T$ and vanishes on $T^c$, then $\vA\vD\vx=0\Rightarrow \vD\vx=0$ by (b). So we get $\vD_T\vx'=\vD\vx=0$.

 \vspace{0.1in}
(d)$\Rightarrow$(b) Assume $\vz\in \ker \vA\cap \vD\Sigma_{2s}$, so we can write $\vz$ as $\vz=\vD\vx_1-\vD\vx_2$
where $\vx_1,\vx_2$ are both $s$-sparse. Let $T$ be the support of $\vx_1-\vx_2$ and  $|T|\leq 2s$. Now if we let $\vx'$ be a vector of length $|T|$ that is just the truncate of $\vx_1-\vx_2$ on $T$, then we have $\vA\vD_T\vx'=0$. The assumption of (d) tells us $\ker \vA\vD_T\subset\ker \vD_T$, which means $\vD_T\vx'=0$. This is equivalent to $\vz=0$.
\qed

A natural question to ask is what is the minimum number of measurements needed so that the equivalent conditions in Proposition \ref{pro_basic} hold. One can imagine that this depends on the specific frame. In the case  when a frame is reduced to an orthonormal basis, we know the absolute minimum number is $2s$~\cite{Simon}.

A surprising result is that the number of rows required of $\vA$ is also $2s$ if the spark of the frame is at least $2s+1$. The \emph{spark} of a frame $\vD$ is  the smallest cardinality of the set of linearly dependent columns from $\vD$. We refer interested readers to \cite[Theorem IV.2.5]{chen}  for its proof. We choose not to elaborate the details here because as far as applications goes, the theoretical lower bound $2s$ cannot be achieved. For one reason, the decoder \pdz\ is not computationally feasible, therefore we solve the $\ell_1$ relaxation of \pdz\ instead, which is  exactly the proposed \ls\ method.

\section{Why \ls?}\label{sec_why}
Within the content of Basis pursuit, there are two major approaches to the frame-based compressed sensing: the \ls\ method and the $\ell_1$-analysis method. The present paper focuses on the synthesis approach, and we would like to briefly explain why.  We also refer \cite{NDEG} for more discussion on these two methods.

First, as has been pointed out in Section \ref{sec_basic}, the synthesis approach is the $\ell_1$ relaxation of \pdz, therefore a very natural method to start with for the frame-based compressed sensing problem. Second, the \la\ approach requires that $\vD^*\vz_0$ to be sparse, instead of the more general setting that $\vz_0$ is $\vD$-sparse. 
It is argued in \cite{Dic_Candes} that with certain condition on $\vD$, $\vz_0$ is $\vD$-sparse will imply $\vD^*\vz_0$ is sparse, however, this imposes more restrictions on $\vD$. Third, it appears that the \la\ method is a sub-problem of \ls\ by the work of Li et al.~\cite{LML12}, which we will elaborate below.

In \cite{LML12}, a frame-based sparse signal is reconstructed by solving the \emph{optimal dual based $\ell_1$-analysis} problem:
\begin{equation}\label{equ_ana_dual}
\hat{\vz}=\arg\min_{\vz\in \F^d, \vD\tilde\vD^*=\vI} \|\tilde\vD^*\vz\|_1, \quad\text{ s.t. } \|\vy-\vA \vz\|_2\leq \epsilon.
\end{equation}

To compare this minimization with the \la\ \eqref{equ_ana}, we see that \eqref{equ_ana} fixes the dual to be the canonical dual whereas \eqref{equ_ana_dual} searches through all feasible $\vz$'s and all dual frames of $\vD$. The idea behind this can be explained by introducing a concept called sparse duals. For any $\vz_0=\vD\vx_0$ with $\|\vx_0\|_0\leq s$, one can always find a dual frame of $\vD$, denoted as $\tilde \vD$, such that $\vx_0 = \tilde\vD^*\vz_0$. This $\tilde\vD$ is called the \emph{sparse dual} of $\vz_0$. It depends on $\vz_0$, and does not need to be unique.

The $\ell_1$-analysis method \eqref{equ_ana} can only find signals that are sparse in the canonical dual, whereas the optimal dual based $\ell_1$-analysis will find the appropriate sparse dual, and hence get a more accurate reconstruction.

Theorem 2 in \cite{LML12} shows that the optimal dual based $\ell_1$-analysis is equivalent to $\ell_1$-synthesis, which implies that \ls\ is a more thorough method than \la.
An efficient split Bregman algorithm is used to solve \eqref{equ_ana_dual} in \cite{LML12} and the numerical experiments also suggest that \ls\ is more accurate.

\section{The null space property for a frame}\label{sec:nsp}
This section explores appropriate conditions on the sensing matrix $\vA$ such that
the frame-based sparse signals can be reconstructed exactly via $\ell_1$-synthesis method \eqref{equ_PD}.

One obvious solution is to require $\vA\vD$ to satisfy the RIP condition or NSP condition. However, as we mention earlier, this leads to an accurate recovery of the coefficients $\vx_0$, which is unnecessary. Moreover, either condition will inevitably require the frame $\vD$ to be incoherent. Therefore, our goal is to find conditions weaker than $\vA\vD$ having RIP or NSP, that will recover the signal accurately, but not necessarily the coefficients, and that will allow coherent frames.

In basis-based compressed sensing, NSP is an equivalent condition for recovering exactly sparse signals under no noise. Therefore we start with a null space property like condition with the hope to characterize the successful recovery of noise-less \ls.


\begin{definition}[Null space property of a frame $\vD$ ($\vD$-NSP)]
Fix a dictionary $\vD\in \F^{d\times n}$, a matrix $\vA\in \F^{m\times d}$ is said to satisfy the $\vD$-NSP of order $s$ ($s$-$\vD$-NSP) if for any index set $T$ with $|T|\leq s$, and any $\vv\in \vD^{-1}(\ker \vA\backslash\{0\})$, there exists $\vu\in\ker \vD$, such that
\begin{equation}\label{eq: DNSP}\|\vv_T+\vu\|_1<\|\vv_{T^c}\|_1.\end{equation}
\end{definition}

The intuition of $\vD$-NSP comes from the fact that we are only interested in recovering $\vz_0$ rather than the representation $\vx_0$.
Comparing this new condition with $\vA\vD$ having NSP, the major difference is that here the left hand side of \eqref{eq: DNSP} is essentially the minimum of $\|\vv_T+\vu\|_1$ over all $\vu\in \ker \vD$, whereas $\vu$ has to be 0 in the latter condition. Therefore this new condition is weaker than $\vA\vD$ having NSP.

 Notice a different DNSP was introduced in \cite{ACP11}, but that version is for the \la\ model.



The following theorem asserts that $s$-$\vD$-NSP is a necessary and sufficient condition for  \eqref{equ_PD} to  successfully recover all the $\vD$-sparse signals with sparsity $s$. 

\begin{theorem}\label{thm_iff}
Fix a frame $\vD\in \F^{d\times n}$, a matrix $\vA\in \F^{m\times d}$ satisfies $s$-$\vD$-NSP if and only if for any $\vz_0\in \vD\Sigma_s$, we have $\hat \vz=\vz_0$, where $\hat \vz$ is the reconstructed signal from $\vy=\vA \vz_0$ using the $\ell_1$-synthesis method \eqref{equ_PD}.
\end{theorem}
\proof Suppose that $\ell_1$-synthesis is successful for all the signals in $\vD\Sigma_s$. Take any support $T$ such that $|T|\le s$ and $\vv\in \vD^{-1}(\ker \vA/\{0\})$. Let $\vz_0=\vD\vv_T$ be the signal that we are trying to recover, then by assumption, the minimizer must be $\vv_T+\vu$ with some $\vu\in \ker \vD$.

Note that $\vv_T-\vv$ is also feasible for \eqref{equ_PD}, that is, $\vA\vD(\vv_T-\vv)=\vA\vD \vv_T$. Then it cannot be a minimizer since $\vD(\vv_T-\vv)\neq \vD\vv_T$, therefore
$$\|\vv_T+\vu\|_1<\|\vv_T-\vv\|_1=\|\vv_{T^c}\|_1.$$

On the other hand, if $\vA$ satisfies $s$-$\vD$-NSP, suppose by contradiction that ($\text{P}_{\vD}$) is not successful for some  $\vz_0=\vD\vx_0$ with $\|\vx_0\|_0\leq s$. This implies that some minimizer $\hat\vx$ of  \eqref{equ_PD} satisfies $\vD\hat \vx\neq \vD\vx_0$ and $\|\hat \vx\|_1<\|\vx_0\|_1$. Let $\vv=\vx_0-\hat \vx$, then $\vv \in \vD^{-1}(\ker \vA/\{0\})$. Take $T$ to be the support of $\vx_0$, then there exists $\vu\in \ker \vD$ such that $\|\vv_T+\vu\|_1<\|\vv_{T^c}\|_1$, i.e. $\|\vx_0-\hat \vx_T+\vu\|_1<\|\hat \vx_{T^c}\|_1$, so
$$\|\vx_0+\vu\|_1\leq\|\vx_0-\hat \vx_T+\vu\|_1+\|\hat \vx_T\|_1<\|\hat \vx_{T^c}\|_1+\|\hat \vx_T\|_1=\|\hat \vx\|_1.$$
This contradicts to the fact that  $\hat \vx$ is a minimizer.\qed

One can easily define a \nspD\ with the $\ell_q$ quasi-norm and generalize the above theorem  for the $\ell_q$-synthesis method for $0<q\leq1$. However, this paper will focus on the $\ell_1$ case since the analysis is  the same for the $\ell_q$ case and the content adds very little to the work.

Notice that when $\vD$ is the canonical basis of $\F^d$, the \nspD\ is reduced to the normal \nsp\ with the same order. In other words, $\vD$-NSP is a generalization of \nsp\ for the frame-based compressed sensing. It is, however, a nontrivial generalization.

For a fixed frame $\vD$, we just need to choose a sensing matrix $\vA$ such that it satisfies \nspD. Obviously different frames will impose different \nspD\ on the sensing matrix. For certain frames, it is possible that no reasonable sensing matrices satisfy \nspD\ at all. We introduce the concept of admissibility.

\begin{definition}[Admissibility]\label{def_adm}
Given a frame $\vD\in \F^{d\times n}$,
We call a frame $\vD$  \emph{$s$-admissible to $\vA$} if $\vA$ has $s$-\nspD.
We call a frame $\vD$ \emph{$s$-admissible} if there exists a sensing matrix with rank strictly less than $d$ (non-trivial) that is $s$-\nspD.
We call a frame $\vD$ \emph{$s$-inadmissible} if it is not $s$-admissible.
\end{definition}
Admissibility is \nspD\ from the perspective of frames. It is the least condition that the frame needs to have if the \ls\ approach is taken. Hence it is expected that highly coherent frames can be admissible. 

\begin{proposition}\label{prop_sub}
If $\vD=[\vB, \vv]$ where $\vB$ is a full rank $d\times (n-1)$ matrix and $\vv=\vB\va$ with $\|\va\|_1\leq 1$ and $\|\va\|_0=k$, then

(1) If $\vA$ has \nspD, then $\vA$ has $\vB$-\nsp\ with the same sparsity $s$.

(2) If $\vA$ has $s$-$\vB$-\nsp, then $\vA$ has \nspD\ with sparsity $s-k+1$. In particular, if $\vv$ is a column of $\vB$,  then $\vA$ has \nspD\ with the same order $s$.

\end{proposition}

\proof Both parts of the proof will use the following argument.

Given $\vx_0, \vy_0$ such that $\vB\vx_0=\vD\vy_0$, consider the following two problems
\begin{equation}\label{equ_PB}
\hat \vx ={\rm arg \thinspace min} \|\vx\|_1, \quad\text{ s.t. } \ \ \ \vA\vB \vx=\vA\vB\vx_0.
\end{equation}
\begin{equation}\label{equ_PDD}
\hat \vy ={\rm arg \thinspace min} \|\vy\|_1, \quad\text{ s.t. } \ \ \ \vA\vD\vy=\vA\vD\vy_0.
\end{equation}

Notice that $[\hat \vx^T,0]^T$ is feasible in \eqref{equ_PDD}, therefore $\|\hat \vy\|_1\leq\|\hat \vx\|_1$.

Write $\hat \vy$ as $\hat \vy^T=[ \vw^T, \alpha]^T$ where $\alpha$ is the $n$-th coordinate of $\hat \vy$, then the vector $[(\vw+\alpha\va)^T, 0]^T$ is feasible in \eqref{equ_PDD} since
\begin{equation}\label{equ_dec}
\vD\hat \vy=\vB\vw+\alpha\vv=\vB\vw+\alpha\vB\va=\vB(\vw+\alpha\va)=\vD([\vw+\alpha\va)^T, 0]^T).
\end{equation}
Moreover,
$$\|[(\vw+\alpha\va)^T, 0]^T\|_1=\|\vw+\alpha\va\|_1\leq \|\vw\|_1+|\alpha|\|\va\|_1\leq \|\vw\|_1+|\alpha|=\|\hat \vy\|_1\leq\|\hat \vx\|_1,$$
indicating that $[(\vw+\alpha\va)^T, 0]^T$ is a minimizer of \eqref{equ_PDD}, and $\vw+\alpha\va$ is a minimizer of \eqref{equ_PB}, hence
\begin{equation}
\|\hat \vx\|_1=\|\vw+\alpha\va\|_1=\|[(\vw+\alpha\va)^T, 0]^T\|_1=\|\hat \vy\|_1
\end{equation}

(1) Suppose $\vA$ has $s$-$\vD$-\nsp, take an arbitrary $s$ sparse vector $\vx_0\in\F^{n-1}$. We solve \eqref{equ_PB} with this $\vx_0$ and \eqref{equ_PDD} with $\vy_0^T=[\vx_0^T, 0]^T\in\F^n$. So $\vB\vx_0=\vD\vy_0$.

Therefore $[\hat \vx^T, 0]^T$ is a minimizer of \eqref{equ_PDD}, and by our assumption and Theorem \ref{thm_iff}, $\vB\hat \vx=\vD([\hat \vx^T, 0]^T)=\vD\vy_0=\vB\vx_0$, thus completes one direction of the proof again by Theorem \ref{thm_iff}.

(2) Suppose $\vA$ has $s$-$\vB$-NSP, take an arbitrary $s-k+1$ sparse vector $\vy_0$, since $\|\va\|_0=k$, there exists $\vx_0$ that is $s$ sparse such that $\vD\vy_0 = \vB\vx_0$ (In fact, we use the same procedure as in \eqref{equ_dec}).

We solve \eqref{equ_PB} and \eqref{equ_PDD} with such $\vx_0$, $\vy_0$. Again write $\hat \vy$ as $\hat \vy^T=[ \vw^T, \alpha]^T$, so we get $\vw+\alpha\va$ is a minimizer of \eqref{equ_PB}, by Theorem \ref{thm_iff}, $\vB(\vw+\alpha\va)=\vB\vx_0$, hence $\vD\hat \vy=\vD\vy_0$, which finishes the proof by again using Theorem \ref{thm_iff}. \qed

A particular case of Proposition \ref{prop_sub} (2) is when $\vv$ is a column of $\vB$, in which case $\|\va\|_0=1$, so both $\vD=[\vB, \vv]$ and $\vB$ are $s$-admissible to $\vA$.

This proposition implies that if a frame is admissible to $\vA$, then a new frame formed by deleting a column or adding a repeated column is still admissible to $\vA$. This is certainly not the case for $\vA\vD$ having RIP or NSP.

\vspace{0.1in}

\subsection{What dictionaries are admissible?}\label{sec_exa}
With Proposition \ref{prop_sub}, we are able to find a class of admissible dictionaries along with their compatible sensing matrices.

First, we already expect random matrices $\vA$ and incoherent frames $\vD$ to be admissible with each other by \cite{DMP1}, because such $\vA\vD$ will have RIP (hence NSP) with very high probability. Second, Proposition \ref{prop_sub} tells us that it is fine to add a repeated column to $\vD$, and the resultant frame is still admissible to the same sensing matrix $\vA$. One can also add a column that is a combination of two frame vectors from $\vD$, but this will reduce the sparsity level by 1. This way we have constructed a class of highly coherent dictionaries that will still be admissible to random sensing matrices. We will discuss more interesting results on this topic later.

\section{Stability guarantees with \nspD}\label{sec_sta}
It is known that the \nsp\ is a sufficient and necessary condition not only for the sparse and noiseless recovery, but also for compressible signals with noisy measurements~\cite{ACP11, Sun11}. In this section  we will  show that this result can generalize to \nspD, but only in the real vector space. For this, we introduce a stronger null space property first.

\begin{definition}[Strong null space property of a frame $\vD$ (\snspD)]
A sensing matrix $\vA$ is said to have \snspD\ of order $s$ if there exists a positive constant $c$ such that for every $ \vv\in \ker(\vA\vD)$,$|T|\leq s$, there exists $\vu\in \ker \vD$ satisfying
\begin{equation}\label{equ_nsp}
\|\vv_{T^c}\|_1-\|\vv_T+\vu\|_1\geq c\|\vD\vv\|_2.
\end{equation}
\end{definition}

This new null space property is obviously no weaker than \nspD, hence the name. With this stronger \nspD, sparse signals can be stably recovered via \PDe\ as follows.

\begin{theorem}\label{thm_sta}
If $\vA$ satisfies $\vD$-SNSP of order $s$, then any solution $\hat \vz$ of \PDe satisfies
$$\|\hat \vz-\vz_0\|_2\leq \frac{2}{c}\sigma_s(\vx_0)+\epsilon\left(\frac{2\sqrt{n}}{c\nu_\vA\nu_\vD}+\frac{2}{\nu_\vA}\right).$$
where $\vx_0$ is any representation of $\vz_0$ in $D$,  $\nu_\vA$ and $\nu_\vD$ are the smallest positive singular values of $\vA$ and $\vD$.
\end{theorem}

The following result will be used multiple times in this paper, so we put it as a lemma whose proof makes use of standard properties of the singular value decomposition. We refer the readers to \cite[Lemma III.4.3]{chen} for the proof.

\begin{lemma}\label{SVD}
Suppose $\vM$ is an $k\times l$ matrix where $k\leq l$, then any vector $\vh\in\F^l$ can be decomposed as $\vh=\va+\vb$ with $\va\in\ker \vM$, $\vb\perp\ker \vM$, and $\|\vb\|_2\leq \frac{1}{\nu_\vM }\|\vM\vh\|_2$, where $\nu_\vM$ is the smallest positive singular value of $\vM$.
\end{lemma}



%


\proof (Proof of Theorem \ref{thm_sta})

Suppose $\vz_0=\vD\vx_0$ is the signal that we want to recover and $T$ is the index set with $k$ largest coefficients in magnitude of $\vx_0$. Let $\vh=\hat \vz-\vz_0=\vD(\hat{\vx}-\vx_0)$, and by Lemma \ref{SVD} we can decompose $\vh$ as $\vh=\vD\vw+\vg$, where $\vD\vw\in\ker \vA$ and  $\|\vg\|_2\leq \frac{1}{\nu_\vA}\|\vA\vh\|_2\leq\frac{2\epsilon}{\nu_\vA}$ with $\nu_\vA$ being the smallest positive singular value of $\vA$.

Use Lemma \ref{SVD} again with the matrix $\vD$, we can find $\vf$ such that $\vg=\vD\vf$, and
\begin{equation}\label{equ_xi}
\|\vf\|_2\leq\frac{1}{\nu_\vD}\|\vg\|_2\leq\frac{2\epsilon}{\nu_\vA\nu_\vD}.
\end{equation}

Since $\vD(\hat{\vx}-\vx_0)=\vh=\vD(\vw+\vf)$, then $\hat{\vx}-\vx_0=\vw+\vf+\vu_1$ with some $\vu_1\in \ker \vD$.
Let $\vv=\vw+\vu_1$, then $\hat{\vx}-\vx_0=\vv+\vf$ and we have $\vD\vv=\vD\vw\in\ker \vA$.

Since $\vv\in \ker(\vA\vD)$, by the definition of \snspD, there exists $\vu\in\ker \vD$ such that \eqref{equ_nsp} holds. Therefore
\begin{align}\notag
&\|\vv+\vx_{0,T}\|_1-\|-\vu+\vx_{0,T}\|_1\\\notag
\geq &\|\vv_{T^c}\|_1+\|\vv_T+\vx_{0,T}\|_1-\|-\vu_T+\vx_{0,T}\|_1-\|\vu_{T^c}\|_1\\\notag
\geq &\|\vv_{T^c}\|_1-\|\vv_T+\vu_T\|-\|\vu_{T^c}\|_1\\\label{equ1}
=&\|\vv_{T^c}\|_1-\|\vv_T+\vu\|_1\geq c\|\vD\vv\|_2.
\end{align}
On the other hand, from the fact that $\hat{\vx}$ is a minimizer, we get
\begin{align*}
&\|-\vu+\vx_{0,T}\|_1+\|\vx_{0,T^c}\|_1\\
\geq&\|-\vu+\vx_0\|_1\geq \|\hat{\vx}\|_1= \|\vv+\vx_0+\vf\|_1 \\
\geq&\|\vv+\vx_0\|_1-\|\vf\|_1\\
\geq&\|\vv+\vx_{0,T}\|_1-\|\vx_{0,T^c}\|_1-\|\vf\|_1.
\end{align*}
Rearrange the above inequality, we get
\begin{equation}\label{equ2}
\|\vv+\vx_{0,T}\|_1-\|-\vu+\vx_{0,T}\|_1\leq 2\|\vx_{0,T^c}\|_1+\|\vf\|_1.
\end{equation}
Combining \eqref{equ1} and \eqref{equ2}, we obtain
\begin{equation}\label{equ_dv}
\|\vD\vv\|_2\leq\frac{2}{c}\|\vx_{0,T^c}\|_1+\frac{1}{c}\|\vf\|_1\leq\frac{2}{c}\|\vx_{0,T^c}\|_1+\frac{\sqrt{n}}{c}\|\vf\|_2.
\end{equation}
In the end, \eqref{equ_dv} and \eqref{equ_xi} together imply,
\begin{align}\notag
\|\vh\|_2&=\|\vD\vv+\vD\vf\|_2\leq\|\vD\vv\|_2+\|\vg\|_2\leq\frac{2}{c}\|\vx_{0,T^c}\|_1+\frac{\sqrt{n}}{c}\|\vf\|_2+\frac{1}{\nu_\vA}2\epsilon\\\label{equ_const}
&\leq\frac{2}{c}\|\vx_{0,T^c}\|_1+\frac{2\sqrt{n}}{c\nu_\vA\nu_\vD}\epsilon+\frac{1}{\nu_\vA}2\epsilon.
\end{align}

The theorem then follows from the definition of $T$.
\qed

Similar to \nspD, the following proposition says that  \snspD\ also allows perfectly correlated frames.

\begin{proposition}\label{pro_snsp}
If $\vA$ has \snspD, then  $\vA$ has $[\vD, \vv]$-\snsp\ with the same sparsity $s$ where $\vv$ is any column of $\vD$.
\end{proposition}
\proof Without loss of generosity, assume $\vv$ is the first column of $\vD$. Let $\vD_1=[\vD,\vv]$.

Take any $\vw\in \ker \vA\vD_1$, and  suppose $a=\vw_{n+1}$ is the $(n+1)$-th coordinate of $\vw$,
$$\vD_1\vw=\vD([\vw_1+a, \vw_2, \cdots, \vw_n]^T).$$

Denote $\bar \vw=[\vw_1+a, \vw_2, \cdots, \vw_n]^T$, hence $\bar \vw\in\ker \vA\vD$.

Take any support $T\subset \{1, 2, \cdots, n+1\}$ whose cardinality is at most $s$, we consider the following 4 cases

Case 1: $1\in T, n+1\not\in T$

By our assumption, there exists $\bar \vu\in \ker \vD$ such that
$\|\bar \vw_{T^c}\|_1-\|\bar \vw_T+\bar \vu\|_1\geq c\|\vD\bar \vw\|_2$. Let $\vu=[\bar \vu^T, 0]^T\in\ker \vD_1$, so
\begin{align*}
\|\vw_{T^c}\|_1-\|\vw_T+\vu\|_1\ge\|\bar \vw_{T^c}\|_1+|a|-\|\bar \vw_T+\bar \vu\|_1-|a|\ge c\|\vD\bar \vw\|_2= c\|\vD_1\vw\|_2.
\end{align*}

Case 2: $1\notin T, n+1\not\in T$

By our assumption, there exists $\bar \vu\in \ker \vD$ such that
$\|\bar \vw_{T^c}\|_1-\|\bar \vw_T+\bar \vu\|_1\geq c\|\vD\bar \vw\|_2$. Let $\vu=[\bar \vu^T, 0]^T\in\ker \vD_1$, so
\begin{align*}
\|\vw_{T^c}\|_1-\|\vw_T+\vu\|_1\ge\|\bar \vw_{T^c}\|_1-\|\bar \vw_T+\bar \vu\|_1\ge c\|\vD\bar \vw\|_2= c\|\vD_1\vw\|_2.
\end{align*}

Case 3: $\{1, n+1\}\in T$

Let $S= T\backslash \{n+1\}$, by our assumption, there exists $\bar \vu\in \ker \vD$ such that
$\|\bar \vw_{S^c}\|_1-\|\bar \vw_S+\bar \vu\|_1\geq c\|\vD\bar \vw\|_2$. Let $\vu=[\bar \vu^T, 0]^T\in\ker \vD_1$, $\va:=[a,0,\cdots, 0, -a]^T$ also belongs to $\ker \vD_1$.
\begin{align*}
&\|\vw_{T^c}\|_1-\|\vw_T+\va+\vu\|_1=\|\bar \vw_{S^c}\|_1-\|[(\bar \vw_S)^T, 0]^T+ \vu\|_1\\
=&\|\bar \vw_{S^c}\|_1-\|\bar \vw_S+ \bar \vu\|_1\geq c\|\vD\bar \vw\|_2= c\|\vD_1\vw\|_2.
\end{align*}

Case 4: $1\not\in T, n+1\in T$

Let $S=(T\backslash\{n+1\})\cup{1}$, by our assumption, there exists $\bar \vu\in \ker \vD$ such that
$\|\bar \vw_{S^c}\|_1-\|\bar \vw_S+\bar \vu\|_1\geq c\|\vD\bar \vw\|_2$. Let $\vu=[\bar \vu^T, 0]^T\in\ker \vD_1$, $\va:=[a,0,\cdots, 0, -a]^T$ also belongs to $\ker \vD_1$. Denote $\vb=\vw_{\{1\}}$ (only keep the first component, set others to zero), then
\begin{align*}
&\|\vw_{T^c}\|_1-\|\vw_T+\va+\vu\|_1\ge\|\bar \vw_{S^c}\|_1+|\vw_1|-\|\vw_T+\va+\vb+\vu-\vb\|_1\\
\geq &\|\bar \vw_{S^c}\|_1+|\vw_1|-\|\vw_T+\va+\vb+\vu\|_1-\|\vb\|_1\\
=&\|\bar \vw_{S^c}\|_1-\|\bar \vw_S+\bar \vu\|_1\geq c\|\vD\bar \vw\|_2= c\|\vD_1\vw\|_2.
\end{align*}
\qed


One may wonder how strong this newly introduced \snspD\  is, or how much stronger it is than \nspD. The following theorem claims that these two conditions are completely the same if we are in a real vector space.


\begin{theorem}\label{thm_equiv}
A matrix $\vA$ satisfying \nspD\ is equivalent to $\vA$ satisfying \snspD\ with the same order when $\vA\in\R^{m\times d}$ and $ \vD\in\R^{d\times n}$.
\end{theorem}

 We postpone the proof to the last section because it is rather involved.
Briefly speaking, we need the function $f(\vv):=(\|\vv_{T^c}\|_1-\|\vv_T+\vu\|_1)/\|\vD\vv\|_2$ to be bounded away from zero. In the case when $\vD$ is a basis, the justification is simple since we can restrict $f$ to a compact set $\ker \vA\backslash\{0\}\cap\Sb^{d-1} $ and apply the extreme value theorem. This is essentially the key step in the stability analysis of \cite{ACP_Samp}.
 In the general frame case $\vD^{-1}(\ker \vA\backslash\{0\})\cap \Sb^{n-1}$ is no longer a compact set, therefore other constructions to overcome this difficulty are necessary.
We conjecture that this theorem is also true for the complex signals, but our proof techniques cannot be generalized to that case.

This theorem indicates that in a real vector space, the \nspD\ is  also  sufficient to guarantee the stable recovery of almost frame-sparse signals under the perturbation of measurements.


\section{\nspD, \snspD, and NSP}\label{sec_gap}
In this section, we explore the differences between $\vA$ having \nspD, $\vA$ having \snspD, and $\vA\vD$ having NSP. In the real case, we have the first two conditions equate to each other, and the third condition being stronger, but we have not addressed  how much stronger.

In general, it is expected that $\vA\vD$ having NSP is stronger than $\vA$ having \snspD. In fact, the following proposition provides a stronger statement.
\begin{proposition}\label{fullcoro}
If $\vA\vD$ has NSP, then there exists a positive constant $c$ such that for every $\vv\in\ker (\vA\vD)$,
\begin{equation}\label{full2}
\|\vv_{T^c}\|_1-\|\vv_T\|_1\geq c\|\vv\|_2.
\end{equation}
Consequently, $\vA\vD$ having NSP is stronger than $\vA$ having \snspD.
\end{proposition}
\proof
If $\vv=0$, then \eqref{full2} is true. So we can assume that $\vv\neq0$.

We define a function
\[
f(\vv)=\frac{\|\vv_{T^c}\|_1-\|\vv_T\|_1}{\|\vv\|_2}
\]
on the set $\ker \vA\vD\setminus\{0\}$. This function is strictly positive by definition of NSP.

 The continuous $f(\vv)$ attains its minimum on the compact set $B=\{\vv: \vv\in \ker \vA\vD, \ \|\vv\|_2=1\}$, i.e. $f(\vv)\geq c>0$ for $\vv\in B$. For any $\vv\in\ker \vA\vD\setminus\{0\}, \frac{\vv}{\|\vv\|_2}\in B$, so $$f(\vv)=f(\frac{\vv}{\|\vv\|_2})\geq c.$$
\qed


Thus, as one would have guessed, the order of these three conditions are
$$\vA\vD \text{ has NSP }\Longrightarrow \vA \text{ has \snspD\ }\Longrightarrow \vA \text{ has \nspD\ }.$$
We would like to view the last two conditions being almost equal, so the question is what exactly is the gap between the first and the third condition. If one looks at the definition of these two conditions, the superfluous difference is that the third condition has a flexibility in the null space of $\vD$. The following theorem describes the behavior of the null space of $\vD$ as a difference of these two conditions.

\begin{theorem}\label{thm_gap}
Fix a sparsity level $s$, if for any $\vu\in \ker \vD$ and any index set $|T|\leq s$, there exists a $\vtu\in \ker \vD$, such that
\begin{equation*}
\|\vu_T+\vtu\|_1< \|\vu_{T^c}\|_1,
\end{equation*}
then $\vA$ having \nspD\ is equivalent to $\vA\vD$ having NSP of the same order $s$.
\end{theorem}

This theorem is a result of the following two lemmas.
\begin{lemma}
Assume $\vA$ satisfies $s$-\nspD. If in addition, for any $\vu\in \ker \vD$ and any index set $|T|\leq s$, there exists a $\vtu\in \ker \vD$, such that
\begin{equation}\label{equ_ker}
\|\vu_T+\vtu\|_1< \|\vu_{T^c}\|_1,
\end{equation}
then for any $\vv\in\ker(\vA\vD)$ and any index set $|T|\leq s$, there exists $\vu\in\ker \vD$ such that
\begin{equation}\label{full3}
\|\vv_{T^c}\|_1-\|\vv_T+\vu\|_1\geq c\|\vv\|_2.
\end{equation}
\end{lemma}
\proof
If $\vv=0$, then \eqref{full3} is true by choosing $\vu=0$. So we can assume  that $\vv\neq0$.

Let
\[
f(\vv)=\sup\limits_{\vu\in \ker \vD}\frac{\|\vv_{T^c}\|_1-\|\vv_T+\vu\|_1}{\|\vv\|_2}.
\]
be a function defined on the set $\ker \vA\vD\setminus\{0\}$. The argument is rather similar to that of Proposition \ref{fullcoro}. The definition of $s$-$\vD$-NSP and \eqref{equ_ker} imply that $f$ is strictly positive on its domain. The continuous $f(\vv)$ attains its minimum on the compact set $B=\{\vv: \vv\in \ker \vA\vD, \ \|\vv\|_2=1\}$, i.e. $f(\vv)\geq 2c>0$ for $\vv\in B$. For any $\vv\in\ker \vA\vD\setminus\{0\}$, normalizing $\vv$ to $\frac{\vv}{\|\vv\|_2}\in B$ completes the proof.
\qed


Notice that \eqref{full3} is a very strong property. It is stronger than \snspD. Not only will \eqref{full3} lead to a stable reconstruction of the signal, it also guarantees the accurate reconstruction of the representation $\vx_0$.

\begin{lemma}\label{lem_sta}
If \eqref{full3} is satisfied, then
 any minimizer $\hat \vx$ of \PDe\ satisfies
$$\|\hat\vx-\vx_0\|_2\leq \frac{2}{c}\sigma_k(\vx_0)+2\nu_{\vA\vD}\epsilon.$$
\end{lemma}

\proof
 The proof of this Lemma is rather similar to that of Theorem \ref{thm_sta}.
Define $\vh=\tilde\vx-\vx_0$, then $\|\vA\vD\vh\|_2\le 2\epsilon$. Decompose $\vh=\vv+\vg$ where $\vv\in\ker\vA\vD$ and $\vg\perp\ker\vA\vD$. Therefore by Lemma \ref{SVD}, $\|\vg\|_2\le\frac{2\epsilon}{2\nu_{\vA\vD}}$.

The assumption \eqref{full3} implies there exists $\vu\in\ker \vD$ such that
$$\|\vv_{T^c}\|_1-\|\vv_T+\vu\|_1\geq c\|\vv\|_2. $$

The fact that $\hat \vx$ is a minimizer indicates
$$\|-\vu+\vx_{0,T}\|_1+\|\vx_{0,T^c}\|_1\geq\|-\vu+\vx_0\|_1\geq\|\hat{\vx}\|_1= \|\vv+\vx_0\|_1
\geq\|\vv+\vx_{0,T}\|_1-\|\vx_{0,T^c}\|_1.$$

Combining above and the same argument as in \eqref{equ1}, we arrive at
$$c\|\vv\|_2\le2\|\vx_{0,T^c}\|_1.$$
hence
$$\|\vh\|_2\le\|\vv\|_2+\|\vg\|_2\le\frac{2\|\vx_{0,T^c}\|_1}{c}+\frac{2\epsilon}{2\nu_{\vA\vD}}.$$
\qed

The conclusion of Lemma \ref{lem_sta} is equivalent to $\vA\vD$ having NSP by~\cite{ACP_Samp}.

\section{Full spark case}\label{sec:fullspark}

This section focuses on an  important question for the frame-based compressed sensing: Can highly coherent dictionaries be admissible? We will see that the null space property plays a key role.

It may seem at first that the answer is positive because we have constructed highly coherent admissible frames in Section \ref{sec_exa}. However, these examples only represent a small class of frames. The key feature the frames have is that there exist small number of columns that are linearly dependent, in which case we say the frames have small spark (see end of Section \ref{sec_basic}). Let us look at the following example where the frame has very big spark (in fact, it achieves the maximum spark $d+1$). 

\begin{example}\label{exa}
Suppose $\vI$ is the identity matrix in $\R^d$, and $\vD$ is a frame formed by concatenating $\vI$ with another vector $\vw = [1+\epsilon, \epsilon, \cdots, \epsilon]^T$. Assume $\epsilon>0$ is small so that  $\vw$ is strongly correlated with the first column of $\vI$. It is easy to see that $\text{ker}(\vD)$ is one dimensional and generated by $\vu=(\vw^T,-1)^T$. We assume $\epsilon$ is small enough such that for some index set $|T|\geq 2$ containing the first and the last indices, and for any $\vu\in\ker\vD$, $\|\vu_T\|_1>\|\vu_{T^c}\|_1$.  We now show that such a frame is $|T|$-inadmissible.

Assume to the contrary that there exists a non-trivial sensing matrix $\vA$ satisfying $|T|$-\snspD, hence $\vD^{-1}(\ker\vA\backslash\{0\})$ is not empty.   Choose $\vv\in \vD^{-1}(\ker\vA\backslash\{0\})$, and an $\alpha$ such that $\alpha\|\vu\|_{\min} >\|\vv\|_\infty$. The \snspD\ of $\vA$ implies that for the vectors $\vv+\alpha \vu$, $-\vv+\alpha \vu$ $\in \vD^{-1}(\ker\vA\backslash\{0\})$, there exist $c_1$, $c_2\in \mathbb{R}$ such that
\[
\|\vv_T+\alpha \vu_T-c_1\vu\|_1<\|\vv_{T^c}+\alpha \vu_{T^c}\|_1,
\]
and
\[
\|-\vv_T+\alpha \vu_T-c_2 \vu\|_1<\|-\vv_{T^c}+\alpha \vu_{T^c}\|_1.
\]
 Adding up the above two inequalities,  we get
\begin{align}\label{eq:add}
\|\vv_T+\alpha \vu_T-c_1 \vu\|_1+\|-\vv_T+\alpha \vu_T-c_2 \vu\|_1
<\|\vv_{T^c}+\alpha \vu_{T^c}| _1+\|-\vv_{T^c}+\alpha \vu_{T^c}\|_1.
\end{align}
The right hand side of \eqref{eq:add} equals $2\alpha\|\vu_{T^c}\|_1$ by our choice of $\alpha$. The left hand side can be bounded from below:
\begin{align*}
&\|\vv_T+\alpha \vu_T-c_1 \vu\|_1+\|-\vv_T+\alpha \vu_T-c_2 \vu\|_1 \\
\geq& \|2\alpha \vu_T-c_1\vu-c_2\vu\|_1\\
=& |2\alpha-c_1-c_2|\|\vu_T\|_1+(|c_1+c_2|)\|\vu_{T^c}\|_1.
\end{align*}
Therefore
$$|2\alpha-c_1-c_2|\|\vu_T\|_1 <(2\alpha-|c_1+c_2|)\|\vu_{T^c}\|_1\leq |2\alpha-c_1-c_2|\|\vu_{T^c}\|_1,$$
or simply $\|\vu_T\|_1 <\|\vu_{T^c}\|_1$, which is a contradiction to our assumption on $\vu$.
\end{example}

Example \ref{exa} is not the most encouraging news for finding coherent and admissible frames. But one can still hope that this is due to the special construction of this frame and perhaps we are still able to find large class of highly coherent and admissible frames with a different structure. The following theorem says otherwise. In fact, its proof can be viewed as a generalization of Example \ref{exa}.

Before stating the theorem, we introduce a special kind of frame. A finite frame of $\C^d$ is called \emph{full spark} if its spark reaches the maximum value $d+1$. In other words, every $d$ columns of this matrix are linearly independent. In particular, the frame $\vD$ in Example \ref{exa} is full spark.

\begin{theorem}\label{thm_main}
The following conditions are equivalent if $\vD$ is full spark,
\begin{enumerate}
\item[(a)] $\vA\vD$ has $s$-\nsp;
\item[(b)] $\vA$ has $s$-\snspD;
\item[(c)] $\vA$ has $s$-\nspD.
\end{enumerate}
\end{theorem}

\begin{remark}
Full spark is not a strong assumption on frames. In fact, it is quite obvious that if we randomly choose the entries of $\vD$ according to any continuous distribution, then with probability 1 we will get a full spark dictionary. Also in ~\cite{BCM12}, it is proved that
full spark Parsevel frames are dense in the space of all Parsevel frames, and a large class of full spark Harmonic frames is also constructed in~\cite{BCM12}.
\end{remark}

Theorem \ref{thm_main} is a corollary of the following lemma and Theorem \ref{thm_gap}.

 \begin{lemma}\label{fullspark}
 If $\vA$ satisfies $s$-$\vD$-NSP and $\vD$ is full spark, then for any index set $T$ with $|T|\le s$ and any $\vu\in \ker \vD$, there exists a $\tilde{\vu}\in \ker \vD$, such that
\begin{equation}\label{full1}
\|\vu_T+\tilde{\vu}\|_1< \|\vu_{T^c}\|_1.
\end{equation}
\end{lemma}
\proof  To rule out the trivial case, suppose $\ker \vA\neq \emptyset$.

First, we will show that for an index set $T$ with $|T|<d$, and any $\vu\in \ker \vD$, there exists $\vv\in \vD^{-1}(\ker \vA \backslash\{0\})$, such that
$\supp \vu_{T^c} \subset \supp \vv_{T^c}$.

As a matter of fact, since $\spark(\vD)=d+1$ and $\vu\in \ker \vD$, we have $|\supp \vu|\geq d+1$, and thus $|\supp \vu_{T^c}| \geq d+1-|T|>d-|T|$. Take any index set $G\subset \supp \vu_{T^c}$ with $|G|=d-|T|$, then $|G\cup T|=d$.  Let $\vD_{G\cup T}$ be the submatrix of $D$ corresponding to the index set $G\cup T$. Then $\vD_{G\cup T}$ is full rank by the full spark assumption. On the other hand, $\ker \vA\neq \emptyset$ implies $\vD^{-1}(\ker \vA\backslash \{0\})\neq \emptyset$. Assume $\vv_0$ is an element in this nonempty set. Let $\vv$ be the vector defined by $\vv_{(G\cup T)^c}=0$ and $\vv_{G\cup T}=\vD^{-1}_{G\cup T}\vD\vv_0$. Then obviously we have $\vD\vv=\vD\vv_0$, $\supp \vv_{T^c}\subset \supp \vu_{T^c} $ and $\vv\in \vD^{-1}(\ker \vA\backslash \{0\})$.

Choose $\alpha $ big enough such that $\alpha\|\vu\|_{\min}>\|\vv\|_\infty$. Since $\vA$ satisfies $s$-$\vD$-NSP, there exist $\vu_1$, $\vu_2\in \ker \vD$, such that
\[
\|(\vv+\alpha \vu)_T+\vu_1\|_1<\|(\vv+\alpha \vu)_{T^c}\|_1,
\]
and
\[
\|(-\vv+\alpha \vu)_T+\vu_2\|_1<\|(-\vv+\alpha \vu)_{T^c}\|_1.
\]
Adding the above two inequalities, and by the choice of $\alpha$, we get
\[
\|2\alpha \vu_T+(\vu_1+\vu_2)\|_1 <2\alpha \|\vu_{T^c}\|_1,
\]
which implies \eqref{full1}.
\qed


As we mentioned, full spark frames are dense and represent a large collection of frames. Hence for ``most" frames $\vD$, if we want to find a sensing matrix $\vA$ such that $\vD$ is admissible to $\vA$, we have to require the composite $\vA\vD$ to have NSP. While $\vA\vD$ having NSP may not directly imply the low coherence of $\vD$,  as far as practice is concerned, this is basically imposing an incoherence condition on $\vD$ for two reasons: First of all, $\vA\vD$ having NSP means $\vD$ having NSP, which, after some analysis (see Section \ref{sec:cohproof}), implies that 
\begin{equation}\label{equ:mu}
\mu(\vD)<1-\frac{2A^2}{nB},
\end{equation}
 where $A, B$ are the frame bounds of $\vD$. Admittedly, the right hand side of \eqref{equ:mu} is not as small as we hoped. Second of all,  current techniques suggest that to verify whether $\vD$ has NSP,  the most efficient way is to require $\vD$ to have RIP, which implies a small incoherence of $\vD$.

It is not yet clear what happens when the frame is of low spark. Of course Section \ref{sec_exa} provides examples of highly coherent admissible frames with low spark, but we do not yet have a general criteria for admissible frames. Nevertheless, some simulations are performed in Section \ref{sec_sim1} for low spark frames. The results are very positive in the sense that low spark frames are very likely to be admissible even with high coherence.

As a consequence, to solve a frame-based compressed sensing problem, if one takes the $\ell_1$-synthesis approach, an incoherent frame $\vD$ is required since $\vD$ is very likely to be full spark. One can argue that in the case of a non-full-spark frame, $\vD$ may still be allowed to be highly coherent. However, due to the denseness of full spark frames, a perturbation on a frame can easily turn non-full-spark to full spark, which again falls into the case that $\vD$ needs to be incoherent. We refer the work by \cite{DMP1} on quantitative bounds for how incoherent a frame needs to be.


\section{Recovery performance with inadmissible frames}\label{sec_coh}

The denseness of full spark frames leads to a very interesting phenomenon: the admissibility of a frame is not stable with respect to perturbation. In Example \ref{exa}, If the last column of the above $\vD$ were identical to the first column, by Proposition \ref{prop_sub}, this frame would have been perfectly fine. However, the small perturbation $\epsilon$ leads to a completely opposite situation.

But the good news is that even with an inadmissible frame, as long as it is close to an admissible one, the reconstructed error is very small.  This is due to the fact that the solution of \PDe\ is stable to the perturbation of $\vD$.  Some related results can also be found in \cite{ACP11}, but with the $\ell_1$-analysis method. Please see Section \ref{sec_sim} for some numerical experiments.

In what follows, $\|\vM\|_{1\rightarrow 2}$ denotes the norm of $\vM$ as an operator from $\ell_1$ to $\ell_2$.

\begin{theorem}\label{thm:robust}
Let $\vz_0=\vD_0\vx_0$ be the true signal with an $s$-sparse representation $\vx_0$. Let $\vy=\vA\vz_0+\vw=\vA\vD_0\vx_0+\vw$ be the noisy measurement with known noise level $\|\vw\|_2 <\epsilon$. Suppose there exists a frame $\vD$ lying close to $\vD_0$: $\|\vD_0-\vD\|_{1\rightarrow 2}\leq \delta$, and $\vD$ is $s$-admissible to $\vA$.  then we have the following stability result for the reconstruction of $\vz_0$ from $\vy$.
\begin{enumerate}[I.]
\item Let $\hat{\vz}$ be the solution of \PDe\ with the frame $\vD$ and $\epsilon$ replaced respectively by $\vD_0$ and $\rho=2\delta\|\vA\|_2\|\vx_0\|_1+\epsilon$.
Then we have
\begin{equation*}\label{eq:result}
\|\hat{\vz}-\vz_0\|\leq 2\delta\|\vx_0\|_1+\frac{2\sqrt n}{c\nu_\vA\nu_\vD}\rho+\frac{2\rho}{\nu_\vA}.
\end{equation*}

\item If $\vx_0$ is also a minimizer of
$$\min\|\vx\|_1,  \quad\text{ s.t. } \vD\vx=\vD\vx_0,$$
(this is quite likely to hold since we assume $\vx_0$ to be $s$-sparse). Then if  $\hat{\vz}$ is the solution of \PDe\ with the frame $\vD_0$ and the original $\epsilon$, we have
\[
\|\hat{\vz}-\vz_0\|_2\leq 2\delta\|\vx_0\|_1+\frac{2\sqrt n}{c\nu_\vA\nu_\vD}\rho+\frac{2\rho}{\nu_\vA},
\]
with the same definition for $\rho$.

\end{enumerate}
\end{theorem}
\begin{remark}
\text{}
\begin{itemize}
\item Theorem \ref{thm:robust} considers the effect of the perturbation in the frame $\vD$. The case when there is a perturbed measurement matrix $\vA$ has been studied before \cite{ACP11, Strohmer}.     

\item We note that the result of Theorem \ref{thm:robust} is dependent on individual signals rather than being universal for the set $\vD_0\Sigma_s$. In fact, as the numerical experiments have suggested, there seems to be no universality in this case.
\end{itemize}
\end{remark}
We need the following lemma to prove Theorem \ref{thm:robust}.
\begin{lemma}\label{lm:stability}
Given $\vx\in\Sigma_s$, $\Psi\in \F^{d,n}$, and $\Phi \in \F^{m,d}$ with $\Phi$ having \text{$s$-$\Psi$-NSP}. Let $\tilde{\vx}$ be such that
\begin{enumerate}
\item[(A1)]\label{it:1} $\|\tilde{\vx}\|_1\leq \|\vx+\vu\|_1$ for all $\vu\in \ker \Psi$;
\item[(A2)]\label{it:2} $\|\Phi\Psi \tilde{\vx}-\Phi\Psi\vx\|\leq 2\tilde{\epsilon}$.
\end{enumerate}
Then
\[
\|\Psi \tilde{\vx}-\Psi \vx\|_2\leq \frac{2\sqrt n}{c\nu_\Phi\nu_\Psi}\tilde{\epsilon}+\frac{2\tilde{\epsilon}}{\nu_\Phi}.
\]

\end{lemma}
The conclusion of this Lemma is the same as that of Theorem \ref{thm_sta}, except that $\vx$ here is exactly sparse. Notice in Theorem \ref{thm_sta}, the assumption that $\hat \vx$ is a minimizer is actually too strong, because in the proof, we only rely on the fact that $\hat{\vx}$ satisfies (A1) and (A2).

\begin{proof} [Proof of Theorem \ref{thm:robust}]

\textbf{Part I.} We prove this part by showing that the two assumptions of Lemma \ref{lm:stability} are fulfilled if the parameters $(\Phi, \Psi, \vy, \vx,\tilde{\vx},\tilde{\epsilon})$ in that lemma are set to $(\vA,\vD,\vA\vD\vx_0,\vx_0,\hat{\vx},\rho)$.

For assumption (A1), we need to show that $\|\vx_0+\vu\|_1\geq \|\hat{\vx}\|_1$ for all $\vu\in \ker \vD$. Indeed, for any $\vu\in \ker \vD$, if $\|\vx_0+\vu\|_1\geq \|\vx_0\|_1$, then $ \|\vx_0+\vu\|_1\geq \|\vx_0\|_1\geq \|\hat{\vx}\|_1$ by the definition of $\hat{\vx}$; if $\|\vx_0+\vu\|_1<\|\vx_0\|_1$, then it must be $\|\vu\|_1 \leq 2\|\vx_0\|_1$, so
\[
\|\vA(\vD-\vD_0)\vu\|_2\leq 2 \|\vA\|_2\|\vD-\vD_0\|_{1\rightarrow 2}\|\vx_0\|_1\leq 2\delta \|\vA\|_2\|\vx_0\|_1,
\]
and
\begin{align*}
&\|\vA\vD_0(\vu+\vx_0)-\vy\|_2\leq\|\vA\vD_0\vx_0-\vy\|_2+\|\vA\vD_0\vu\|_2=\|\vA\vD_0\vx_0-\vy\|_2+\|\vA(\vD-\vD_0)\vu\|_2\\
\leq & \epsilon+2\delta\|\vA\|_2\|\vx_0\|_2= \rho.
\end{align*}
Hence $\vu+\vx_0$ is feasible in the minimization problem, which implies $\|\hat{\vx}\|_1\leq \|\vx_0+\vu\|_1$.

For assumption (A2), we need to show $\|\vA\vD\hat{\vx}-\vA\vD\vx_0\|_2\leq 2\rho$. First observe that by definition, we have
$$\|\vA\vD_0\hat{\vx}-\vA\vD_0\vx_0\|_2\leq \|\vA\vD_0\hat{\vx}-\vy\|_2+\|\vA\vD_0\vx_0-\vy\|_2\leq \rho+ \epsilon.$$
Then we can calculate that
\begin{align}\label{eq:error}
\|\vA\vD \hat{\vx}-\vA\vD \vx_0\|_2&=\|\vA(\vD-\vD_0)\hat{\vx}-\vA(\vD-\vD_0)\vx_0+\vA\vD_0(\hat{\vx}-\vx_0)\|_2 \notag\\
&\leq \|\vA(\vD-\vD_0)\hat{\vx}\|_2+\|\vA(\vD-\vD_0)\vx_0\|_2+\|\vA\vD_0(\hat{\vx}-\vx_0)\|_2\notag \\
&\leq 2\|\vx_0\|_1\|\vA\|_2\|\vD-\vD_0\|_{1\rightarrow 2}+\epsilon+\rho \notag \\
&=2\delta\|\vA\|_2\|\vx_0\|_1+\epsilon+\rho=2\rho
\end{align}
Now that both assumptions of Lemma \ref{lm:stability} are satisfied, we apply the lemma to obtain
\[
\|\vD\hat{\vx}-\vD\vx_0\|_2\leq \frac{2\sqrt{n}}{c\nu_\vA\nu_\vD}\rho+\frac{2\rho}{\nu_\vA}.
\]
Hence,
\begin{align*}
\|\hat{\vz}-\vz_0\|_1=\|\vD_0\hat{\vx}-\vD_0\vx_0\|_2&\leq \|\vD\hat{\vx}-\vD\vx_0\|_2+\|(\vD-\vD_0)\hat{\vx}\|_2+\|(\vD-\vD_0)\vx_0\|_2  \\
&\leq 2\delta \|\vx_0\|_1+\frac{2\sqrt{n}}{c\nu_\vA\nu_\vD}\rho+\frac{2\rho}{\nu_\vA}.
\end{align*}
This completes the proof of Part I.\\ \\
\textbf{Part II.} Set the parameters  $(\Phi, \Psi, \vy, \vx, \tilde{\vx}, \tilde{\epsilon})$ of Lemma \ref{lm:stability} to $(\vA, \vD, \vA\vD\vx_0, \vx_0, \hat{\vx}, \epsilon)$. For assumption (A1), the additional assumption on $\vx_0$ implies that
 $\|\vx_0\|_1\leq \|\vx_0+\vu\|_1$ for all $\vu\in \ker \vD$. This together with the fact that  $\vx_0$ is feasible to the minimization problem indicates that $\|\hat{\vx}\|_1\leq \|\vx_0\|_1\leq \|\vx_0+\vu\|_1$ for all $\vu\in \ker \vD$.
 For assumption (A2), we can follow exactly the same argument as in Part I.

Applying the lemma, we get
\[
\|\vD\hat\vx-\vD\vx_0\|_1\leq \frac{2\sqrt{n}}{c\nu_\vA\nu_\vD}\rho+\frac{2\rho}{\nu_\vA}.\]
Therefore
\begin{align*}
\|\hat{\vz}-\vz_0\|_1=\|\vD_0\hat{\vx}-\vD_0\vx_0\|_2&\leq \|\vD\hat{\vx}-\vD\vx_0\|_2+\|(\vD-\vD_0)\hat{\vx}\|_2+\|(\vD-\vD_0)\vx_0\|_2  \\
&\leq 2\delta \|\vx_0\|_1+\frac{2\sqrt{n}}{c\nu_\vA\nu_\vD}\rho+\frac{2\rho}{\nu_\vA}.
\end{align*}

\end{proof}


%

\section{Simulations}\label{sec_sim}
\subsection{Reconstruction with low spark frames}\label{sec_sim1}

We construct a $200\times 400$ frame as follows:
$\vD=[\vF,\vG]$, where $\vF$ is the discrete cosine transform matrix with dimension $200\times200$, and $\vG$ is another $200\times 200$ matrix whose columns are linear combinations of 3 columns of $\vF$. In particular, each column of $\vG$ is in the form of
\[
\frac{a_1\vf_{k_1}+a_2\vf_{k_2}+a_3\vf_{k_3}}{\|a_1\vf_{k_1}+a_2\vf_{k_2}+a_3\vf_{k_3}\|_2},
\]
where $a_1, a_2, a_3$ are $N(0,1)$ random variables, and $\vf_{k_1}, \vf_{k_2}, \vf_{k_3}$ are 3 columns of $\vF$ chosen uniformly at random.

With such reconstruction, $\vD$ is a highly coherent unit norm frame with low spark ($\leq 4$). We want to test whether such a $\vD$ is admissible to the Gaussian sensing matrix.

Let  $\vA$ be the $80\times 200$ random Gaussian matrix. The $\ell_1$-magic is used for reconstruction with a tolerance level of $10^{-6}$.
For each sparsity level $s$, we generate random signals $\vz_0=\vD\vx_0$ with $\vx$ being $s$-sparse having Gaussian entries. We run it 500 times, and take the largest relative reconstruction error of signals among  500 trials, which is denoted as $E_{\vz,s}=\max_{500\text{ trials }} \frac{\|\hat \vz-\vz_0\|_2}{\|\vz_0\|_2}$.

The first row of Table 1 shows the values of $E_{\vz,s}$ for various $s$.
As we see, the first three errors of the first row are at the tolerance level and thus can be considered as 0, indicating that $\vD$ is admissible to $\vA$ with $s\leq 8$. 
This empirical result suggests that when a frame is low spark, a highly coherent frame can still be admissible, unlike the full spark case.
But this experiment cannot be explained by Proposition \ref{prop_sub} since we have added 200 columns. A future direction of research could be developing more theoretical results on low spark admissible frames.

\subsection{Reconstruction with inadmissible frames}

With a small perturbation to the above $\vD$, we can get highly coherent frames that is of full spark with probability 1. In particular, we still let $\vD=[\vF,\vN]$, where $\vF$ is the same but the columns of $\vN$ draw from the following distribution:
 \begin{equation}\label{eq:distr}
\frac{a_1\vf_{k_1}+a_2\vf_{k_2}+a_3\vf_{k_3}+\epsilon \vg}{\|a_1\vf_{k_1}+a_2\vf_{k_2}+a_3\vf_{k_3}+\epsilon \vg\|_2},
\end{equation}
where $a_1, a_2, a_3$ are $N(0,1)$ random variables, $\vf_{k_1}, \vf_{k_2}, \vf_{k_3}$ are 3 columns of $\vF$ chosen uniformly at random, $\epsilon$ is some small positive number, and $\vg$ is a $200\times 1$ random Gaussian vector so that the term $\epsilon\vg$ works as a perturbation.

We again run $\ell_1$-magic 500 times with randomly generated $\vD$-sparse signals. Row 2-5 of Table 1 shows how the worst reconstruction error among 500 trials $E_{\vz,s}$ changes with $\epsilon$. All the errors are bigger than the tolerance level. Therefore for each fixed sparsity level and $\epsilon$, some signals are not considered to be reconstructed, which implies $\vD$ is not admissible to $\vA$.  This has been predicted by Theorem \ref{thm_main}: coherent and full spark dictionaries are not admissible.

However, we can see that as $\epsilon$ approaches  0, $E_{\vz,s}$ is getting smaller and smaller, indicating a smaller error when the frame is approaching  an admissible one.  This is supported by Theorem \ref{thm:robust}.

We have also attached Table 2, which lists the biggest error of the coefficients $E_{\vz,s}=\max_{500\text{ trials }} \frac{\|\hat \vx-\vx_0\|_2}{\|\vx_0\|_2}$ for the same frame and sparsity in comparison with Table 1. The errors are all very big due to the high coherence of $\vA\vD$. But when $\epsilon=0$ and sparsity level is low, this does not prevent the accurate reconstruction of the signal.

\begin{table}[ht]
\caption{Maximum reconstruction error of signals over 500 trials}
\centering

\begin{tabular}{c c c c c c c c}
\hline
\hline
$\epsilon$ & $\mu(\vD)$ & $E_{\vz,2}$ & $E_{\vz,5}$ & $E_{\vz,8}$ & $E_{\vz,11}$ & $E_{\vz,14}$ & $E_{\vz,17}$\\
\hline
0 & 0.9999 & 2.64$\times 10^{-6}$ & 5.08$\times10^{-6}$ & 7.58$\times10^{-6}$ & $28.76\times10^{-6}$ & 0.15 & 0.32\\
$0.0001$ & 0.9993& $0.001$ & $0.009$ & $0.003$ & $0.002$ & 0.15 & 0.47\\
$0.001$ & 0.9965 & 0.007 & 0.020 & 0.020 & $0.026$ & 0.11 & 0.38\\
0.003 & 0.9983 & $0.018$ & 0.041 & 0.036 & 0.060 & 0.19 & 0.32 \\
0.009 & 0.9961 & $0.068$ & 0.100 & 0.170 & 0.145 & 0.22 & 0.35 \\
[1ex]
\hline
\end{tabular}
\label{table:nonlin}
\end{table}

\begin{table}[ht]
\caption{Maximum reconstruction error of coefficients over 500 trials}
\centering
\begin{tabular}{c c c c c c c c}
\hline
\hline
$\epsilon$ & $\mu(\vD)$ & $E_{\vx,2}$ & $E_{\vx,5}$ & $E_{\vx,8}$ & $E_{\vx,11}$ & $E_{\vx,14}$ & $E_{\vx,17}$\\
[0.5ex]
\hline
0 & 0.9999             & 1.01& 0.98& 1.05 & 0.99 & 0.89 & 0.92\\
$0.0001$ & 0.9993& 1.04 & 1.24& 1.03 & 0.86& 0.90 & 0.92\\
$0.001$ & 0.9965  & 1.18 & 0.98 & 0.91 & 1.09 & 0.95 & 0.96\\
0.003 & 0.9983      & 0.65 & 1.03 & 0.88 & 1.00 & 0.85 & 0.82 \\
0.009 & 0.9961      & 1.18 & 0.80 & 0.99 & 0.66 & 0.72 & 0.81 \\
[1ex]
\hline
\end{tabular}
\label{table:nonlin}
\end{table}

\section{Some proofs}
\subsection{NSP and incoherence}\label{sec:cohproof}
\begin{theorem}
If a unit norm frame $\vD=\{\vd_i\}_{i=1}^n\in\F^{d\times n}$ has $2$-NSP, and has frame bounds $A,B>0$, that is, $A\|\vx\|_2^2\leq \sum_{i=1}^n|\langle \vx,\vd_i\rangle|^2\leq B\|\vx\|^2$, then its coherence satisfies
$$\mu(\vD)<1-\frac{2A^2}{nB}.$$
\end{theorem}
\begin{proof}
WLOG, we assume $|\langle \vd_1,\vd_2\rangle|=\mu$, where $\vd_1$ and $\vd_2$ are the first two columns of $\vD$.  Let $\vx=\vD^T(\vD\vD^T)^{-1}(\vd_1-\sgn(\langle \vd_1,\vd_2\rangle)\vd_2)$. It is easy to verify that
\[
\vD(\vx+\ve)=0,  \ve=[-1, \sgn(\langle \vd_1,\vd_2\rangle,0,...,0]^T,
\]
and
\[
\|\vd_1- \sgn(\langle \vd_1,\vd_2\rangle)\vd_2\|^2=2-2\mu.
\]
Since $\vx+\ve\in \ker(\vD)$ and $\vD$ has $2$-NSP, we know that
\[\|(\vx+\ve)_{\{1,2\}}\|_1\leq \|(\vx+\ve)_{\{3,...,n\}}\|_1,\]
which implies 
\[ 
2\leq \|\vx\|_1.
\]
On the other hand, by the definition of $\vx$, we have
\[
\|\vx\|_1\leq \sqrt n\|\vx\|_2\leq \sqrt{n}\frac{\sqrt B}{A}\sqrt{2-2\mu}.
\]
The last two equations together imply the conclusion of the theorem.
\end{proof}
\subsection{Proof of Theorem \ref{thm_equiv}}\label{sec:proof}
We divide this proof into three Lemmas, of which the last two are the main content of the proof. Moreover, our technique is not directly applicable to complex vector spaces mainly because that the additivity of $\ell_1$ norm in complex spaces is different from that of the real case.

\begin{lemma}\label{lem_comp}
Fix a dictionary $\vD\in \R^{d\times n}$, suppose the measurement matrix $\vA\in \R^{m\times d}$ satisfies $s$-$\vD$-NSP and $T$ is an index set with cardinality at most $s$. Define
$$h(\vw)=\sup_{\vtu\in\ker \vD}\frac{\|\vw_{T^c}\|_1-\|\vw_T+\vtu\|_1}{\|\vD\vw\|_2},$$
then $h(\vw)$ has positive lower bound on the set
$$W=\{\vw: \vw\in \vD^{-1}(\ker \vA\backslash\{0\}), \|\vw\|_2\leq C_1\|\vD\vw\|_2\},$$
where $C_1$ is a positive constant such that $W$ is not empty.
\end{lemma}
\proof First, it is easy to see that $h(\vw)>0$ on $W$ because $\vA$ has $\vD$-NSP. $h$ is also continuous on $W$ since  $\sup_{\vtu\in\ker \vD}-\|\vw_T+\vtu\|_1=-\inf_{\vtu\in\ker \vD}\|\vw_T+\vtu\|_1$, which is continuous.

Note that $W\cap  {\Sb^{n-1}}=\ker (\vA\vD)\cap \Sb^{n-1}\cap \{\|\vw\|_2\leq C_1\|\vD\vw\|_2\}$ is a non-empty compact set, therefore $h(\vw)\geq C_2$ on  $W\cap \Sb^{n-1}$ for some positive constant $C_2$.

For any $\vw\in W$, since $h(\vw/\|\vw\|)\geq C_3$, there exists $\vtu\in\ker \vD$ such that
$$\frac{\|\vw_{T^c}/\|\vw\|_2\|_1-\|\vw_T/\|\vw\|_2+\vtu\|_1}{\|\vD\vw\|_2/\|\vw\|_2}>C_2/2,$$
which implies
$$\frac{\|\vw_{T^c}\|_1-\|\vw_T+\vtu\cdot\|\vw\|_2/\|_1}{\|\vD\vw\|_2}>C_2/2,$$
hence $h(\vw)> C_2/2.$ \qed

\vspace{0.1in}

Fix a support $T$, a vector $\vv\in \vD^{-1}(\ker \vA\backslash\{0\})$,  define
$$g_\vv(\vu,t)=\sup_{\vtu\in\ker \vD}\left(\|(t\vv+\vu)_{T^c}\|_1-\|(t\vv+\vu)_T+\vtu\|_1\right)$$
and
 $$f_\vv(\vu,t)=g_\vv(\vu,t)/t$$ for $\vu\in \ker \vD$ and $t\ge0$.
Note that the fact $\vA$ satisfies $\vD$-NSP implies that $g_\vv(\vu,t)>0$ and $f_\vv(\vu,t)>0$ for any $(\vu,t)$ in the domain.

\begin{lemma}\label{lem1}
For any fixed $\vv\in \vD^{-1}(\ker \vA\backslash\{0\})$ and index set $T$, suppose  $\{\vu_i\}_{i=1}^{\infty}$, $\{t_i\}_{i=1}^{\infty}$, $\{\vb_i\}_{i=1}^{\infty}$, and $ \{\vw_j\}_{i=1}^{L}$ satisfy 

(1) $\vu_i\rightarrow \vu_0, t_i>0, t_i\rightarrow0$,
and $\lim_{i\rightarrow\infty}f_\vv(\vu_i, t_i)=0$,

(2)  $\vb_i=\vu_i-\vu_0=\sum_{j=1}^L\beta_i(j)\vw_j$ with $\beta_i(j)\ge0$ , and $\vb_i, \vw_i$ are in the same orthant of $\R^n$.

(3)  $\sgn(\vw_j(k)+\vu_0(k))=\sgn(\vu_0(k))$ for all $k\in \supp(\vu_o)$ and $1\le j\le L$,

then there must be a coordinate $j_0\in \{1, 2, \dots, L\}$ such that $\frac{\beta_i(j_0)}{t_i}\not\rightarrow\infty$.
\end{lemma}
\proof Suppose by contradiction that $\frac{\beta_i(j)}{t_i}\rightarrow\infty$ for every coordinate $j$.

 Choose $K$ big enough whose value will be specified later. Set $c_i(j)=\beta_i(j)-\frac{t_i}{t_K}\beta_K(j)$,  so $\vb_i-\frac{t_i}{t_K}\vb_K= \sum c_i(j)\vw_j$. By our assumption,  $c_i(j)>0$ when $i$ is big enough.

Given any $\epsilon>0$, by definition of supremum, there exist $\vtu_1, \vtu_2$, and $\vtu_3\in\ker D$ such that
\begin{align}\nonumber
&\sum_j c_i(j)g_\vv(\vw_j+\vu_0,0)+\frac{t_i}{t_K}g_\vv(\vu_K,t_K)+(1-\sum_j c_i(j)-\frac{t_i}{t_K})g_{\vv}(\vu_0,0)\\\nonumber
\leq&\sum_j c_i(j)\left[\|(\vw_j+\vu_0)_{T^c}\|_1-\|(\vw_j+\vu_0)_{T}+\vtu_1\|_1\right]+\epsilon\\\nonumber
&+\frac{t_i}{t_K}\left[\|(t_K\vv+\vb_K+\vu_0)_{T^c}\|_1-\|(t_K\vv+\vb_K+\vu_0)_{T}+\vtu_2\|_1\right]+\epsilon\\\nonumber
&+ (1-\sum_j c_i(j)-\frac{t_i}{t_K})[\|(\vu_0)_{T^c}\|_1-\|(\vu_0)_{T}+\vtu_3\|_1]+\epsilon\\\label{equ_sgn}
=& \|\left[\sum_j c_i(j)(\vw_j+\vu_0)+\frac{t_i}{t_K}(t_K\vv+\vb_K+\vu_0)+ (1-\sum_j c_i(j)-\frac{t_i}{t_K})\vu_0\right]_{T^c}\|_1+3\epsilon\\\nonumber
&-\sum_j c_i(j)\|(\vw_j+\vu_0)_{T}+\vtu_1\|_1-\frac{t_i}{t_K}\|(t_K\vv+\vb_K+\vu_0)_{T}+\vtu_2\|_1-\|(\vu_0)_{T}+ (1-\sum_j c_i(j)\\\nonumber
&-\frac{t_i}{t_K})\vtu_3\|_1\\\label{equ_epsilon}
\leq &g_\vv(\vu_i,t_i)+3\epsilon,
\end{align}

\eqref{equ_epsilon} is due to the triangle inequality, and \eqref{equ_sgn} will be justified later.
Now let $\epsilon\rightarrow0$ in \eqref{equ_epsilon}, we get
$$g_\vv(\vu_i, t_i)\geq \frac{t_i}{t_K}g_\vv(\vu_K,t_K)\Rightarrow f_\vv(\vu_i,t_i)\geq f_\vv(\vu_K, t_K),$$
which contradicts to the first assumption.

The rest of this proof is to justify \eqref{equ_sgn}. Due to the fact that  $c_i(j)>0, \frac{t_i}{t_K}>0$, and $1-\sum c_i(j)-\frac{t_i}{t_K}>0$ (if $i$ is big enough), a sufficient condition for \eqref{equ_sgn} to hold is that for each $k\in T^c$, the signs of $\{\vw_j(k)+\vu_0(k)\}_{j=1}^L, t_K\vv(k)+\vb_K(k)+\vu_0(k), \text{ and } \vu_0(k)$ are all the same. This indeed holds because we can choose $K$ such that
$$\frac{\beta_K(j)}{t_K}>\frac{|\vv(k)|}{\max_j|\vw_j(k)|}, \text{ for all index } k\in T^c.$$

With such choice of $K$, we get $|\vv(k)|<\sum_{j=1}^m|\vw_j(k)|\frac{\beta_K(j)}{t_K}=|\sum_{j=1}^m\vw_j(k)\frac{\beta_K(j)}{t_K}|$ since all $\vw_j$'s are in the same orthant.  Hence
$$\sgn(t_K\vv(k)+\sum_{j=1}^m\beta_K(j)\vw_j(k))=\sgn(\sum_{j=1}^m\beta_K(j)\vw_j(k))=\sgn(\vw_j(k)).$$

If $\vu_0(k)=0$, then $\sgn(\vw_j(k)+\vu_0(k))=\sgn(\vw_j(k))=\sgn(t_K\vv(k)+\vb_K(k)+\vu_0(k))$.

If $\vu_0(k)\neq0$, then $\sgn(\vw_j(k)+\vu_0(k))=\sgn(\vu_0(k))$ by the third assumption. Moreover, $\sgn(t_K\vv(k)+\vb_K(k)+\vu_0(k))=\sgn(\vu_0(k))$  when $K$ is big enough since $t_i\rightarrow0, \vb_i\rightarrow0$.
\qed
\begin{lemma}\label{lem_imp}
For any fixed $\vv\in \vD^{-1}(\ker \vA\backslash\{0\})$ and index set $T$, we have $$\inf_{\vu\in\ker \vD, t>0}f_\vv(\vu,t)>0.$$
\end{lemma}
\proof We first argue that it suffices to prove $\inf_{\|\vu\|=1, \vu\in\ker \vD, t>0}f_\vv(\vu,t)>0$. 

When $\vu\neq0$, $f_\vv(\vu,t)=f_\vv\left(\frac{\vu}{\|\vu\|}, \frac{t}{\|\vu\|}\right)$ and when $\vu=0$, $f_\vv(0,t)=f_\vv(0,1)>0$, so
$$\inf_{\vu\in\ker \vD, t>0}f_\vv(\vu,t)=\min\{\inf_{\|\vu\|=1, \vu\in\ker \vD, t>0}f_\vv(\vu,t), f_\vv(0,1)\}.$$

Suppose by contradiction that 
\begin{equation}\label{equ:con}
\inf_{\|\vu\|=1, \vu\in\ker \vD, t>0}f_\vv(\vu,t)=0,
\end{equation} then there exists a sequence $(\vu_i, t_i)$ with $\|\vu_i\|=1$, such that $\lim_{i\rightarrow\infty}f_\vv(\vu_i, t_i)=0$. 

We will eventually construct specific sequences that satisfy the three assumptions of Lemma \ref{lem1}, then arrive at a contradiction.

We first show that $\{t_i\}$ has a subsequence converging to 0.
Otherwise, we have $t_i\geq t_0>0$ for some $t_0$, which results that
 $$t_i\vv+\vu_i\in W=\{\vw: \vw\in \vD^{-1}(\ker \vA\backslash\{0\}), \|\vw\|_2\leq C_1\|\vD\vw\|_2\},$$  for some constant $C_1$ (depending on $\vv$ which is fixed in this lemma). Indeed,
$$\|t_i\vv+\vu_i\|\leq\|t_i\vv\|+1\le\left\{\begin{array}{lc}\|\vv\|+1\leq\frac{\|\vv\|+1}{t_0\|\vD\vv\|}\|\vD(t_i\vv+\vu_i)\|,&t_i\leq 1\\ t_i(\|\vv\|+1)=\frac{\|\vv\|+1}{\|\vD\vv\|}\|\vD(t_i\vv+\vu_i)\|, &t_i>1\end{array}\right..$$
Applying Lemma \ref{lem_comp}, we get $f_\vv(\vu_i,t_i)=h(t_i\vv+\vu_i)\|\vD\vv\|\geq C\|\vD\vv\|$ which  contradicts to \eqref{equ:con}. Without loss of generality, we assume the original sequence $t_i\rightarrow0$.

We can also assume without loss of generality that $(\vu_i,t_i)\rightarrow (\vu_0,0)$. There must be infinitely many of $\{\vu_i-\vu_0\}$ falling into some (closed)  orthant of $\R^n$, say $O$. Again for convenience of notation,  we assume all terms of $\vb_i:=\vu_i-\vu_0$ belong to $O$. The benefit of staying in the same orthant is that for a fixed coordinate $k$, $\sgn(b_i(k))$  is the same for any  $i$. This is to satisfy the second assumption of Lemma \ref{lem1}.

Let $\{\vw_j\}_{j=1}^{L}$ be the extremal rays of the polyhedral cone $\ker \vD\cap O$, i.e., any vector in $\ker \vD\cap O$ can be expressed as a nonnegative linear combination of  $\{\vw_j\}_{j=1}^m$. We divide each $\vw_j$ by a big enough constant to make the components of $\vw_j(k)$ small enough such that $\sgn(\vw_j(k)+\vu_0(k))=\sgn(\vu_0(k))$ for all $k\in \supp(\vu_0)$ and $1\le j\le m$. This is to satisfy the third assumptions of Lemma \ref{lem1}.

We write $\vb_i=\vu_i-\vu_0=\sum_{j=1}^m\beta_i(j)\vw_j$, where $\beta_i(j)\geq0$ and convergent to 0 as $i\rightarrow\infty$. For a fixed $1\le j\le m$, $\left\{\frac{\beta_i(j)}{t_i}\right\}_{i=1}^{\infty}$ must have a subsequence converging to  a finite constant, or $\infty$. Again we assume without loss of generality that the original sequence $\frac{\beta_i(j)}{t_i}$ converges (to a constant or infinity) for every $j$.

If $\frac{\beta_i(j_0)}{t_i}\rightarrow a_{j_0}(\neq\infty)$ for some $j_0$, then with the triangle inequality,

\begin{align*}
g_{\vv}(\vu_i, t_i)&=\sup_{\vtu\in\ker D}\|(t_i\vv+a_{j_0}t_i\vw_{j_0}+\sum_{j\neq j_0}\beta_i(j)\vw_j+\vu_0+(\beta_i(j_0)-a_{j_0}t_i)\vw_{j_0})_{T^c}\|_1\\
&-\|(t_i\vv+a_{j_0}t_i\vw_{j_0}+\sum_{j\neq j_0}\beta_i(j)\vw_j+\vu_0+(\beta_i(j_0)-a_{j_0}t_i)\vw_{j_0})_{T}+\vtu\|_1\\
&\leq o(t_i)+g_{\vv+a_{j_0}\vw_{j_0}}(\vu_i-\beta_i(j_0)\vw_{j_0},t_i)\leq o(t_i)+g_{\vv}(\vu_i, t_i),
\end{align*} which leads to
\begin{equation*}
\lim_{i\rightarrow\infty}f_{\vv+a_{j_0}\vw_{j_0}}(\vu_i-\beta_i(j_0)\vw_{j_0},t_i)=\lim_{i\rightarrow\infty}f_{\vv}(\vu_i, t_i)=0.
\end{equation*}

In general, take $J=\{j: \frac{\beta_i(j)}{t_i}\rightarrow a_j(\neq\infty)\}$, and we get
$$\lim_{i\rightarrow0}f_{\vv'}(\vu_i',t_i)=0,$$ 
where $\vv'=\vv+\sum_{j\in J}a_j\vw_j, \vu_i'=\vu_i-\sum_{j\in J}\beta_i(j)\vw_j$ by the same argument as above.

Notice that the set of sequences $\{\vu'_i\}_{i=1}^{\infty}$, $\{t_i\}_{i=1}^{\infty}$, $\{\vb'_i=\vu'_i-\vu_0\}_{i=1}^{\infty}$, and $ \{\vw_j\}_{i=1}^{L}$ satisfy the three assumptions of Lemma \ref{lem1}. However, $\vb'_i=\sum_{j\not\in J}\beta_i(j)\vw_j$ with $\frac{\beta_i(j)}{t_i}\rightarrow\infty$ for all $j\not\in J$. This contradicts Lemma \ref{lem1}.
\qed

\vspace{0.2in}

\vspace{0.2in}

\noindent\emph{ Proof of Theorem \ref{thm_equiv}.} Suppose $\vA$ satisfies $s$-$\vD$-NSP, we need to show the function
$$F(\vw)= \sup_{\vtu\in\ker \vD}\frac{\|\vw_{T^c}\|_1-\|\vw_T+\vtu\|_1}{\|\vD\vw\|_2}$$
has a positive lower bound on $\vD^{-1}(\ker\vA\backslash\{0\})$ for every $|T|\leq s$.

Decompose $\vw$ as $\vw=t\vv+\vu$ where $\vu=\rm{P}_{\ker \vD}\vw, t\vv=P_{(\ker \vD)^{\perp}}\vw$, with $\|\vv\|=1$, and $t>0$. Therefore
\begin{equation}
\inf_{\vw\in \vD^{-1}(\ker\vA\backslash\{0\})}F(\vw)=\inf_{\vv\in \ker \vD^{\perp}, \|\vv\|=1}\inf_{\vu\in\ker \vD,t>0}f_\vv(\vu,t)/\|\vD\vv\|.
\end{equation}

By Lemma \ref{lem_imp}, the function $\inf_{\vu\in\ker \vD,t>0}f_\vv(\vu,t)$ is always positive. Since the set $(\ker \vD)^{\perp}\cap \Sb^{n-1}$ is compact, it is sufficient to prove that the function $\inf_{\vu\in\ker \vD,t>0}f_\vv(\vu,t)$ is lower-semi continuous with respect to $\vv$.
\begin{align*}
f_{\vv+\ve}(\vu,t)&=\sup_{\vtu\in\ker \vD}\frac{\|(t\vv+t\ve+\vu)_{T^c}\|_1-\|(t\vv+t\ve+\vu)_T+\vtu\|_1}{t}\\
&\geq \sup_{\vtu\in\ker \vD}\frac{\|(t\vv+\vu)_{T^c}\|_1-\|(t\vv+\vu)_T+\vtu\|_1-\|t\ve\|_1}{t}
\end{align*}

Take the infimum over $\vu,t$ of both sides, we get
$$\inf_{\vu\in\ker \vD,t>0}f_{\vv+\ve}(\vu,t)\geq \inf_{\vu\in\ker \vD,t>0}f_\vv(\vu,t)-\|\ve\|_1,$$
which shows this function is lower-semi continuous.\qed

\section*{Acknowledgments}

This research has been supported by Defense Threat Reduction Agency grant HDTRA1-13-1-0015. The authors also thank Jameson Cahill for helpful conversations related to the material. 

\bibliographystyle{abbrv}
\bibliography{citations_13_08_21.bib}

\end{document}